\documentclass{elsarticle}

\usepackage{lineno,hyperref}
\modulolinenumbers[5]

\usepackage[american]{babel}	%English hyphenation etc.
\usepackage[utf8]{inputenc}		%Input encoding
\usepackage[T1]{fontenc}		%Font encoding
\usepackage[english=american]{csquotes}	%Quotation marks ' '
\usepackage{graphicx}			%Include graphics
\usepackage{epstopdf}
\usepackage{subfig}				%Subfigures
\usepackage{marvosym}			%Special symbols
\usepackage{amsmath}			%Mathmode
\usepackage{amssymb}
\usepackage{amsthm}
\usepackage{mathtools}
\usepackage{color}

\theoremstyle{definition}
\newtheorem{definition}{Definition}%[subsection]
\newtheorem{example}{Example}%[subsection]

\newtheorem*{remark}{Remark}

\theoremstyle{plain}

\newtheorem{theorem}{Theorem}%[section]
\newtheorem{proposition}{Proposition}

\begin{document}

\begin{frontmatter}
\title{Comparing the rankings obtained from two biodiversity indices: the Fair Proportion Index and the Shapley Value}

\author{Kristina Wicke}
\ead{kristina.wicke@uni-greifswald.de}

\author{Mareike Fischer \corref{cor1}}
\ead{email@mareikefischer.de}
\cortext[cor1]{Corresponding author}

\address{Department of Mathematics and Computer Science, University Greifswald, Greifswald, Germany}

\begin{abstract}
The Shapley Value and the Fair Proportion Index of phylogenetic trees have been frequently discussed as prioritization tools in conservation biology. Both indices rank species according to their contribution to total phylogenetic diversity, allowing for a simple conservation criterion. While both indices have their specific advantages and drawbacks, it has recently been shown that both values are closely related. However, as different authors use different definitions of the Shapley Value, the specific degree of relatedness depends on the specific version of the Shapley Value -- it ranges from a high correlation index to equality of the indices. 
In this note, we first give an overview of the different indices. Then we turn our attention to the mere ranking order provided by either of the indices.
We compare the rankings obtained from different versions of the Shapley Value for a phylogenetic tree of European amphibians and illustrate their differences. We then undertake further analyses on simulated data and show that even though the chance of two rankings being exactly identical (when obtained from different versions of the Shapley Value) decreases with an increasing number of taxa, the distance between the two rankings converges to zero, i.e., the rankings are becoming more and more alike. Moreover, we introduce our freely available software package FairShapley, which was implemented in Perl and with which all calculations have been performed.
\end{abstract}


\end{frontmatter}

\section{Introduction}
\label{intro}
Due to limited financial means, biodiversity conservation programs often need to prioritize the species to conserve. Two indices used in this matter are the Shapley Value and the Fair Proportion Index. Both are based on phylogenetic trees and rank species according to their contribution to overall biodiversity.

The Shapley Value was first introduced by \citep{Haake_2007} for unrooted trees and reflects the average biodiversity contribution of a species. The Fair Proportion Index, on the other hand, lacks a biological link to conservation, but is significantly easier to calculate and has been preferred in practice. Under a different name (ED for \emph{Evolutionary Distinctiveness}) the Fair Proportion Index has for example been used in the `EDGE of Existence' Project, established by the \emph{Zoological Society of London} in 2007 (see \citep{edge}). 
However, \citep{Hartmann2013} observed a strong correlation between the Shapley Value and the Fair Proportion Index on rooted trees, where the Shapley Value was calculated for the unrooted version of the tree by suppressing the root vertex. Very recently, \citep{Fuchs2015} have extended the concept of the Shapley Value to rooted trees and have shown that the two indices are identical for these trees. They also introduced a slightly modified version of the Shapley Value, which again is highly correlated to the Fair Proportion Index. 

In this note we first give an overview of the various versions of the Shapley Value and their respective relatedness with the Fair Proportion Index, before we focus on the mere ranking order of taxa obtained from different versions of the indices. Although the indices are highly correlated, they can result in different ranking orders, especially when the trees become large. We will show with a simulation study based on random trees that in fact, despite the increasing correlation as the number of species grows, different ranking orders are still more likely than equal ones. Therefore, in order to demonstrate what the correlation really implies, we treat the ranking lists as vectors and use the so-called Manhattan distance to measure the difference between two rankings suggested by different indices. We then show that the distance between these rankings tends to 0 as the number of species grows.

All calculations in this manuscript were performed using our new software tool FairShapley, which has been made publicly available at \\ http://www.mareikefischer.de/Software/FairShapley.zip. \\ This tool, which was implemented in Perl, is able to calculate all versions of the Shapley Value as well as the Fair Proportion Index as explained in this paper.

\section{Preliminaries}
\label{sec:1}
Before we can present our results, we need to introduce some notation and definitions. Recall that a phylogenetic tree is a connected, acyclic graph, where the leaves are bijectively labelled by some set $X$ of species, which are also often called taxa.
A \emph{rooted} phylogenetic tree is a phylogenetic tree with a designated root node $\rho$. In biology, \emph{binary} phylogenetic trees are of particular importance. A phylogenetic tree is called \emph{unrooted binary} if all internal nodes have degree 3. It is called \emph{rooted binary} if all internal nodes have degree 3 except for one specified root node $\rho$ of degree 2.
Throughout this paper, we always specify whether we are referring to rooted or unrooted trees. When we write $T^u$, this notation refers to an unrooted  phylogenetic tree, whereas $T^r$ always refers to a rooted phylogenetic tree. In both cases, when we refer to the size of a tree, we mean the number $n = |X|$ of taxa, i.e., the number of leaves of the tree under consideration.
Note that a rooted tree can also be turned into an unrooted tree by abolishing the designation of a specified root node. In case of binary phylogenetic trees, a rooted tree can be turned into an unrooted tree by suppressing the root node $\rho$, i.e., by deleting $\rho$ and the two edges adjacent to it and re-connecting the two resulting degree-2 vertices with a new edge. We subsequently elaborate how turning a rooted tree into an unrooted one can change the various diversity indices.

In biodiversity conservation, the {\em phylogenetic diversity} of a set of species plays an important role. This concept captures how diverse or different a set of species is. Mathematically, this requires the trees under consideration to come with edge lengths (e.g., representing evolutionary time since the last common ancestor or substitution rate). Therefore, we assume all edges in the trees to have positive edge lengths assigned to them, and we denote the length of an edge $e$ as $\lambda_e$. Moreover, recall that a rooted tree is called {\em ultrametric} if the path lengths from all leaves to the root are identical. Here, the path lengths are calculated as the sum of all edge lengths on the path from a leaf to the root. The concept of ultrametric trees is also often referred to as the {\em molecular clock hypothesis} in biology.
Note, however, that throughout this paper we do not assume ultrametricity unless stated otherwise.

We are now in the position to formally define phylogenetic diversity, or $PD$ for short. 

\begin{definition} The phylogenetic diversity ($PD$) of a phylogenetic tree is defined as follows:
\begin{enumerate} \item For a rooted phylogenetic tree $T^r$ with leaf set $X$, the $PD^r$ of a subset $S \subseteq X$ of taxa is calculated by summing up the edge lengths of the phylogenetic subtree of $T^r$ containing $S$ and the root (i.e., the sum of branch lengths in the smallest spanning tree in $T^r$ containing $S$ and the root). Thus, the $PD$ of a single taxon is the length of the path from the root to the leaf representing this taxon. 
\item In case of an unrooted phylogenetic tree $T^u$, the unrooted phylogenetic diversity, $PD^u$, of a subset $S \subseteq X$ of taxa is defined as the sum of edge lengths in the minimal spanning tree in $T^u$ connecting those leaves. The $PD$ of a single taxon is defined as 0.
\end{enumerate}
\end{definition}

\noindent Note that in an ultrametric tree, all taxa have the same $PD^r$, and note that if one considers the unrooted version $T^u$ of a rooted tree $T^r$, the $PD$ may decrease due to the different definitions.

\begin{example} \label{ex1} Consider Figure \ref{fig:1}, which depicts trees $T^r$ and $T^u$ on taxon set $X=\{A,B,C,D\}$. Note that here, $T^u$ is the tree you get by suppressing the root of $T^r$. Now consider the highlighted subset $S=\{A,B\}\subseteq X$. The phylogenetic diversity of $S$ can be calculated as follows: $PD^r(S)=1+1+1+1=4$, and $PD^u(S)=1+1=2$. The difference between the two definitions of diversity can be explained by the path of length 2 connecting $S$ with the root, which is disregarded in the unrooted case.  \end{example}

\begin{figure}[htbp]
	\centering
	\includegraphics[scale=0.65]{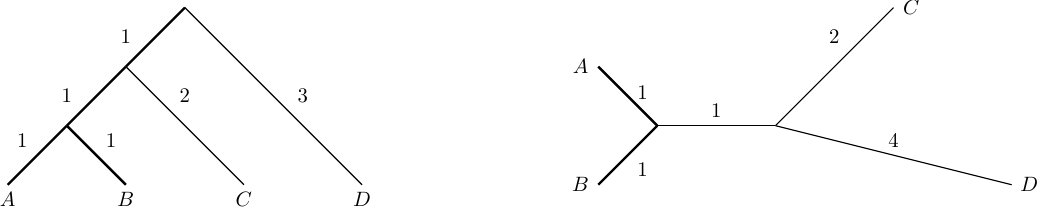}
 	\caption{Rooted tree $T^r$ with leaf set $\{A,B,C,D\}$ and unrooted tree $T^u$, which is derived from $T^r$ by suppressing the root node.}
  \label{fig:1}
\end{figure}

\noindent
One more concept we need before we can turn our attention to diversity prioritization indices is the concept of a \emph{ranking}. Here, a ranking $r$ is just an assignment of ranking numbers to the elements of $X$, where for any pair of taxa $x,y \in X$, $x$ either receives a higher or lower ranking number than $y$ or the ranking numbers of $x$ and $y$ are equal (we then call $x$ and $y$ tied). 
We say that a function $f:X \rightarrow \mathbb{R}$ induces a ranking $r_f$ if the ranking number of $x$ in $r_f$ is smaller than the ranking number of $y$ precisely if $f(x) > f(y)$. If $f(x)=f(y)$ for some $x \neq y$, $x$ and $y$ receive the same ranking number.

\begin{example} Let $X=\{A,B,C,D\}$. Let $f(A)=0.5$, $f(B)=3$, $f(C)=0.2$ and $f(D)=1.5$. Then the induced ranking is $r_f(A,B,C,D) = (3,1,4,2)$.
Now let $g(A)=0.5$, $g(B)=0.5$, $g(C)=0.2$ and $g(D) = 1.5$. Then we retrieve the induced ranking $r_g(A,B,C,D)=(2,2,4,1)$, where $A$ and $B$ are tied.
\end{example}

Next, recall that the so-called Manhattan distance $d_1$ (also known as $L_1$ distance or $l_1$ metric) between two vectors $r,s \in \mathbb{R}^n$ is defined as follows: 
\begin{linenomath*}
$$ d_1(r, s) = \|r - s \| = \sum_{i=1}^{n} \lvert r(i) - s(i) \rvert .$$
\end{linenomath*}
We will later on use the Manhattan distance to measure the difference between two rankings induced by different biodiversity indices. Notice that for comparing rankings, often the so-called Kendall tau distance is used. The Kendall tau distance counts the number of pairwise disagreements between two rankings, but can only deal with total rankings, i.e. rankings without ties. As rankings obtained by different biodiversity indices may include ties, we use the Manhattan distance instead (Comparisons where the Kendall tau distance is used by breaking ties arbitrarily can be found in the supporting information (S1 Text)). \\
However, since we want to observe the behavior of the different prioritization indices for increasing numbers of taxa, we need to normalize the calculated distances. This is due to the fact that whenever the number of taxa increases, even small differences between two rankings have a higher impact on the distance. 
%For example, if you compare the rankings $(1,2)$ and $(2,1)$, their Manhattan distance is $d_1=1+1=2$, and it is maximal in the sense that the two rankings could not be any more different. However, the Manhattan distance of the rankings $(1,2,3,4,5)$ and $(1,2,5,4,3)$ is $d_1= 0+0+2+0+2=4$, and thus higher than the distance between the two rankings mentioned first, even though here only two out of five taxa swapped their places in the ranking. 
So we need to normalize in order to take into account that whenever the number of taxa increases, the maximum possible Manhattan distance increases, too. So we divide exactly by this factor.
Thus, we define the {\em normalized Manhattan distance} $d_1^*(r_1,r_2)$ for two rankings $r_1$ and $r_2$ with associated ranking vectors $v_{r_1}$ and $v_{r_2}$ as follows: 
\begin{linenomath*}
$$ d_1^*(r_1,r_2) \coloneqq \frac{d_1(v_{r_1},v_{r_2})}{\displaystyle \max_{r',s'} \text{ } d_1(v_{r'},v_{s'})}. $$
\end{linenomath*}
Note that the maximum in the denominator is obtained when $r'=(1,2, \ldots, n)$ and $s'=(n,n-1,\ldots,1)$. \\

Now we are in a position to introduce the biodiversity indices, which we will analyze in the following.

\subsection{Various indices for biodiversity conservation}

In this section, we will present and analyze some indices for biodiversity conservation, which have recently been discussed in the literature. All of these indices turn out to be related, but as different authors use different definitions of these indices, their results sometimes differ. We will therefore give an overview about the relationships of the various definitions.

The first index we want to introduce is the {\em Fair Proportion Index}, which is only defined for rooted trees.

\begin{definition}[Fair Proportion Index]
For a rooted phylogenetic tree $T^r$ with leaf set $X$ the Fair Proportion Index of a taxon $a$ is defined as
	\begin{linenomath*}	
	\begin{equation} 
	FP_{T^r}(a) = \sum_{e} \frac{\lambda_{e}}{D_{e}}, \label{def:1}
	\end{equation}
	\end{linenomath*}
where the sum runs over all edges $e$ on the path from $a$ to the root and $D_e$ denotes the number of leaves descendent from that edge.
\end{definition}

It can be easily shown that the sum of all Fair Proportion Indices for a given species set $X$ equals the total branch length of the given tree.

\begin{example} In Figure \ref{fig:1}, $FP_{T^r}(A)=\frac{1}{1}+\frac{1}{2}+\frac{1}{3}=\frac{11}{6}$, $FP_{T^r}(B)=\frac{1}{1}+\frac{1}{2}+\frac{1}{3}=\frac{11}{6}$, $FP_{T^r}( C)=\frac{2}{1}+\frac{1}{3}=\frac{7}{3}$ and $FP_{T^r}(D)=\frac{3}{1}=3$. Altogether, we have $FP_{T^r}(A)+ FP_{T^r}(B)+FP_{T^r}( C )+FP_{T^r}(d)= 9$, which equals the total sum of all branch lengths in $T^r$. Moreover, the ranking induced by the Fair Proportion Index in this case is $r_{FP}(A,B,C,D)=(2,2,4,1)$.
\end{example}

As has been shown in the previous example, the Fair Proportion Index can easily be calculated. However, it does not have a direct biological justification. Therefore, another index from evolutionary game theory was proposed and adjusted to phylogenetic conservation, namely the so-called Shapley Value \citep{Haake_2007, Hartmann2013, Fuchs2015}. However, as various authors use slightly different versions of this index, we present three different definitions here, the first of which we call the original Shapley Value, which can be defined both for rooted and unrooted trees.

\begin{definition}[Original Shapley Value]
Let $T^r$ be a rooted phylogenetic tree with leaf set $X$ and let $PD^r(S)$ denote the phylogenetic diversity of $S \subseteq X$. Then the Shapley Value for a taxon $a \in X$ is defined as
\begin{linenomath*}
\begin{equation}
	 \label{def:2.1}
	SV_{T^r}(a) = \frac{1}{n!}\sum_{\substack{S \subseteq X \\ a \in S}} \Big( (\lvert S \rvert -1)!(n- \lvert S \rvert)! 
	(PD^r(S)-PD^r(S \setminus \{a\})) \Big),
\end{equation}
\end{linenomath*}
where $n = \lvert X \rvert$ and $S$ denotes a subset of species containing taxon $a$ (also sometimes referred to as `coalition') and the sum runs over all such coalitions possible. \\
Similarly, for an unrooted tree $T^u$ with leaf set $X$ we have
	\begin{linenomath*} 
	\begin{equation}
	 \label{def:2.2}
	SV_{T^u}(a) = \frac{1}{n!}\sum_{\substack{S \subseteq X \\ a \in S}} \Big( (\lvert S \rvert -1)!(n- \lvert S \rvert)! (PD^u(S)-PD^u(S \setminus \{a\})) \Big).
	\end{equation}
	\end{linenomath*}
\end{definition}

Note that the definition of the original Shapley Value is basically the same for rooted and unrooted trees. The only difference is how the phylogenetic diversity of subsets is defined (i.e., $PD^r$ vs. $PD^u$). 
%For rooted trees, however, this value coincides with the Fair Proportion Index, i.e., $SV_{T^r} = FP_{T^r}$, which was recently shown by \citep{Fuchs2015}.
For rooted trees, however, the original Shapley Value and the Fair Proportion Index coincide:
\begin{theorem}[\citet{Fuchs2015}] \label{fuchs_theorem}
Let $T^r$ be a rooted phylogenetic tree with leaf set $X$. Then we have for all $a \in X:$
$$SV_{T^r}(a) = FP_{T^r}(a).$$
\end{theorem}

We now present an example for the original Shapley Value.

\begin{example} We calculate $SV_{T^r}(A)$ for the tree depicted in Figure \ref{fig:1} (a). Note that the possible subsets $S \subseteq X=\{A,B,C,D\}$ which contain $A$ are: $\{A\}$, $\{A,B\}$, $\{A,C\}$, $\{A,D\}$, $\{A,B,C\}$, $\{A,B,D\}$, $\{A,C,D\}$ and $\{A,B,C,D\}$. Thus, we have to consider 8 summands when calculating $SV_{T^r}(A)$: 
\begin{linenomath*}
\begin{align*}
SV_{T^r}(A)&= \frac{1}{4!}\sum_{S: A \in S} \Big( (\lvert S \rvert -1)!(\lvert X \rvert -\lvert S \rvert)! (PD^r(S)-PD^r(S\setminus \{A\})) \Big) \\
		   &= \frac{1}{4!} \Big[(1-1)!(4-1)!(3-0)  \\
		  	&\qquad {} + (2-1)!(4-2)! \big( (4-3)+(5-3)+(6-3) \big)  \\
		  	&\qquad {} + (3-1)!(4-3)! \big( (6-5)+(7-6)+(8-6) \big) \\
		  	&\qquad {} + (4-1)!(4-4)! (9-8) \Big] \\
		  &= \frac{11}{6}.
\end{align*}
\end{linenomath*}

Note that, as implied by Theorem \ref{fuchs_theorem} (cf. \citep{Fuchs2015}), this value coincides with $FP^r(A)$ as calculated above, but the calculation is much more involved. Moreover, a similar calculation yields $SV_{T^u}(A)=\frac{19}{12}$, which shows that the original Shapley Value for the rooted and unrooted versions of the tree depicted by Figure \ref{fig:1} differ due to the different underlying definitions of phylogenetic diversity.
\end{example}

%%%%%
Additionally to the original Shapley Value, \citep{Fuchs2015} also introduced a modified version of the Shapley Value, which we will call {\em modified Shapley Value} and which we denote by $\widetilde{SV}$ in the following. The difference between the original and the modified versions of the Shapley Value is that the first considers all subsets of taxa which contain a certain taxon, whereas the latter only considers subsets of size at least 2 (i.e., $\lvert S \rvert \geq 2$). 

\begin{definition} [Modified Shapley Value]
Let $T^r$ be a rooted and $T^u$ be an unrooted phylogenetic tree with leaf set $X$ and let $PD^r(S)$ and accordingly $PD^u(S)$ denote the phylogenetic diversity of a subset $S \subseteq X$. Then the modified Shapley Value for a taxon $a$ is defined as
	\begin{linenomath*}
	\begin{equation}
	\label{def:3.1}
	\widetilde{SV}_{T^r}(a) = \frac{1}{n!}\sum_{\substack{S: a \in S \\ \lvert S \rvert \geq 2}} \Big( (\lvert S \rvert -1)!(n- \lvert S \rvert)! (PD^r(S)-PD^r(S \setminus \{a\})) \Big) 
	\end{equation}
	\end{linenomath*}
	in the rooted case and 
	\begin{linenomath*}
	\begin{equation}
	\label{def:3.2}
	\widetilde{SV}_{T^u}(a) = \frac{1}{n!}\sum_{\substack{S: a \in S \\ \lvert S \rvert \geq 2}} \Big( (\lvert S \rvert -1)!(n- \lvert S \rvert)! (PD^u(S)-PD^u(S \setminus \{a\})) \Big)
	\end{equation}
	\end{linenomath*}
	in the unrooted case, where $n = \lvert X \rvert$ and the sum runs over all coalitions $S$ containing taxon $a$ and at least one other taxon.
\end{definition}

Comparing the original and the modified Shapley Value, we have the following relationships:
\begin{linenomath*} 
\begin{align}
	SV_{T^r}(a) &= \widetilde{SV}_{T^r}(a) + \frac{PD^r(\{a\})}{n} \label{eq:1} \text{ and } \\
	SV_{T^u}(a) &= \widetilde{SV}_{T^u}(a) + \frac{PD^u(\{a\})}{n} = \widetilde{SV}_{T^u}(a). \label{eq:2}
\end{align}
\end{linenomath*}
\begin{proof}
We first establish equation \eqref{eq:1}: \\
Consider the contribution of the singleton set $\{a\}$ to $SV_{T^r}(a)$.
By definition this is
\begin{align*}
\frac{1}{n!} \Big( (1-1)! \, (n-1)! \, (PD^r(\{a\}) - PD^r(\emptyset)) \Big) 
&= \frac{(n-1)!}{n!} \, PD^r(\{a\}) \\
&= \frac{PD^r(\{a\})}{n},
\end{align*}
where $PD^r(\{a\}) \neq 0$, because we are considering the rooted version of $PD$ on a rooted tree $T^r$ with positive edge lengths.
As $\widetilde{SV}_{T^r}(a)$ lacks this contribution it is
$$ \widetilde{SV}_{T^r}(a) = SV_{T^r}(a) - \frac{PD^r(\{a\})}{n},$$
or, in other words,
$$ SV_{T^r}(a) = \widetilde{SV}_{T^r}(a) + \frac{PD^r(\{a\})}{n}.$$ 
Equation \eqref{eq:2}, i.e. equality of the original and modified Shapley Value for unrooted trees, follows from \eqref{eq:1} and the fact that $PD^u(S)$ is defined as zero, whenever $S$ contains only one element. In our case $PD^u(\{a\})=0$ and thus,
$$ SV^u(a) = \widetilde{SV}_{T^u}(a) + \frac{0}{n} = \widetilde{SV}_{T^u}(a).$$
\end{proof}

\begin{remark}
In \citep{Fuchs2015}, further results on the expectation of the Fair Proportion Index (and thus, on the original Shapley Value) and the modified Shapley Value and on their correlation are established for random phylogenetic trees under the Yule-Harding model and the uniform model.
\end{remark}

\citep{Hartmann2013} also states a correlation result for the Fair Proportion Index and the Shapley Value, but does not go into the details of the definition of the Shapley Value that he uses. \noindent\citep{Fuchs2015} suggest that the modified Shapley Value was used in \citep{Hartmann2013}. We think, however, that \citep{Hartmann2013} used yet another version of the Shapley Value, namely the original Shapley Value of the unrooted tree derived from the original tree by suppressing the root node. The reason why we think so is that, while \citep{Hartmann2013} does not give a definition of phylogenetic diversity, he does state a definition of the Shapley Value, and the sum there ranges over all subsets containing a certain taxon, not only subsets of size at least 2. But he cannot be using the original Shapley Value for rooted trees, because otherwise his results would have led to the equality of the Shapley Value and the Fair Proportion Index rather than only a strong correlation. This reasoning is also supported by \citep[p. 142]{Book_Steel} and leads to yet another version of the Shapley Value, which we will call the {\em unrooted Shapley Value on rooted trees}, or {\em unrooted rooted Shapley Value} for short, and which we denote by $\widehat{SV}_{T^r}$. 

\begin{definition}
For a rooted phylogenetic tree $T^r$ with leaf set $X$ we retrieve the unrooted Shapley Value on rooted trees of a taxon $a$ as
\begin{linenomath*}
\begin{equation}
	\widehat{SV}_{T^r}(a) = SV_{T^u}(a),
\end{equation}
\end{linenomath*}
where $SV_{T^u}(a)$ is the original Shapley Value of $a$ in the corresponding unrooted tree $T^u$.
\end{definition}
Recall that turning a rooted tree $T^r$ into an unrooted tree $T^u$ causes a change in the definition of phylogenetic diversity (i.e., a shift from $PD^r$ to $PD^u$). Thus, the unrooted Shapley Value on a rooted tree does not necessarily coincide with the original Shapley Value on the rooted tree. The two indices are, however, highly correlated (see \citep{Hartmann2013}). However, note that this analysis is somewhat counterintuitive as the Fair Proportion Index is only defined for rooted trees, and in \citep{Hartmann2013}, only rooted trees are depicted, but for the Shapley Value still the unrooted version of the rooted tree seems to have been used.

However, just as the unrooted Shapley Value on a rooted tree needs not coincide with the original Shapley Value on the rooted tree, it does not necessarily agree with the modified Shapley Value on the rooted tree, but again, the indices are closely related.  \\

Summarizing the above, we have: 
\begin{itemize} \item $SV_{T^r} = FP_{T^r},$ 
\item $SV_{T^r} \neq \widetilde{SV}_{T^r}$, but $SV_{T^r}(a) = \widetilde{SV}_{T^r}(a) + \frac{PD^r(a)}{n}$ for all $a \in X,$
 \item  $SV_{T^r} \neq \widehat{SV}_{T^r}$, but there is a high correlation between the two values (\citet{Hartmann2013}), 
 \item $SV_{T^u} = \widehat{SV}_{T^r}$, where $T^u$ denotes the unrooted version of $T^r$;
 \item $\widehat{SV}_{T^r} \neq \widetilde{SV}_{T^r}$.
\end{itemize}

%Example
\begin{example} Consider again Figure \ref{fig:1}.  As mentioned above, we have  $FP_{T^r}(A) = SV_{T^r}(A) = \frac{11}{6}$. Moreover, we have $\widetilde{SV}_{T^r}(A)=\frac{13}{12}$. For $T^u$ we have $SV_{T^u}(A)=\frac{19}{12}$. \\
Note that $SV_{T^r}(A) \neq \widetilde{SV}_{T^r}(A) \neq SV_{T^u}(A) = \widehat{SV}_{T^r}(A)$.
\end{example}

\begin{remark} 
We finish this section with an observation on the size of the differences between the different versions of the Shapley Value. In fact, we show that these differences can be made arbitrarily large, but also arbitrarily small.
Consider the rooted phylogenetic tree $T_{\varepsilon}^r$ with leaf set $X=\{A,B,C,D\}$ and its unrooted version $T_{\varepsilon}^u$ depicted in Figure \ref{fig:2}. We now calculate the different versions of the Shapley Value for taxon $A$: 
	\begin{align*}
	SV_{T_{\varepsilon}^r}(A) &= \frac{1}{\varepsilon} + \frac{\varepsilon}{2}, \\
	\widetilde{SV}_{T_{\varepsilon}^r}(A) &= \frac{3}{4 \varepsilon} + \frac{\varepsilon}{4}, \\
	\widehat{SV}_{T_{\varepsilon}^r}(A) &=  \frac{3}{4 \varepsilon} + \frac{3 \varepsilon}{4}.
	\end{align*}
Now consider the difference between the original and the modified/unrooted rooted Shapley Value, respectively. For $\varepsilon \rightarrow 0$, these differences tend to infinity:
	\begin{align*}
	\left\vert SV_{T_{\varepsilon}^r}(A) - \widetilde{SV}_{T_{\varepsilon}^r}(A) \right\vert &= \left\vert \frac{1}{4 \varepsilon} + \frac{ \varepsilon}{4} \right\vert \overset{\varepsilon \rightarrow 0}{\longrightarrow} \, \infty, \\
	\left\vert SV_{T_{\varepsilon}^r}(A) - \widehat{SV}_{T_{\varepsilon}^r}(A)  \right\vert &= \left\vert \frac{1}{4 \varepsilon} - \frac{\varepsilon}{4} \right\vert \overset{\varepsilon \rightarrow 0}{\longrightarrow} \, \infty.
	\end{align*}
However, the difference between the unrooted rooted and the modified Shapley Value tends to zero for $\varepsilon \rightarrow 0$: 
	\begin{align*}
	\left\vert \widehat{SV}_{T_{\varepsilon}^r}(A) - \widetilde{SV}_{T_{\varepsilon}^r}(A) \right\vert &= \left\vert \frac{\varepsilon}{2} \right\vert \overset{\varepsilon \rightarrow 0}{\longrightarrow} \, 0.
	\end{align*}
Note that for $\varepsilon \rightarrow \, \infty$ all three differences tend to infinity, while for $\varepsilon \rightarrow 1$, the difference between the original and the unrooted rooted Shapley value of $A$ tends to zero:
\begin{align*}
\left\vert SV_{T_{\varepsilon}^r}(A) - \widehat{SV}_{T_{\varepsilon}^r}(A) \right\vert &= \left\vert \frac{1}{4 \varepsilon} - \frac{\varepsilon}{4} \right\vert \overset{\varepsilon \rightarrow 1}{\longrightarrow} \, 0.
\end{align*}
If we now consider taxon $B$, we have
	\begin{align*}
	SV_{T_{\varepsilon}^r}(B) &= \frac{3 \varepsilon}{2}, \\
	\widetilde{SV}_{T_{\varepsilon}^r}(B) &= \varepsilon, \\
	\widehat{SV}_{T_{\varepsilon}^r}(B) &=  \frac{1}{12 \varepsilon} + \frac{17 \varepsilon}{12}.
	\end{align*}
For $\varepsilon \rightarrow 0$ the difference between the original and the modified Shapley Value tends to zero, while both the difference between the original and the unrooted rooted Shapley Value and the difference between the unrooted rooted and the modified Shapley Value tend to infinity:
	\begin{align*}
	\left\vert SV_{T_{\varepsilon}^r}(B) - \widetilde{SV}_{T_{\varepsilon}^r}(B) \right\vert &= \left\vert \frac{\varepsilon}{2}  \right\vert \overset{\varepsilon \rightarrow 0}{\longrightarrow} \,  0, \\
	\left\vert SV_{T_{\varepsilon}^r}(B) - \widehat{SV}_{T_{\varepsilon}^r}(B) \right\vert &= \left\vert \frac{\varepsilon}{12} - \frac{1}{12 \varepsilon} \right\vert \overset{\varepsilon \rightarrow 0}{\longrightarrow} \, \infty, \\
	\left\vert \widehat{SV}_{T_{\varepsilon}^r}(B) - \widetilde{SV}_{T_{\varepsilon}^r}(B) \right\vert &= \left\vert \frac{1}{12 \varepsilon} + \frac{5 \varepsilon}{12} \right\vert \overset{\varepsilon \rightarrow 0}{\longrightarrow} \, \infty.
	\end{align*}
Again, for $\varepsilon \rightarrow \, \infty$ all three differences tend to infinity. \\
In conclusion, this example shows that we can make the difference between any pair of different versions of the Shapley Value (original, modified and unrooted rooted) arbitrarily large (tending to infinity), but also arbitrarily small (tending to zero). 
\end{remark}

%Figure (trees with epsilon edges))
\begin{figure*}[htbp]
	\centering
	\includegraphics[scale=0.7]{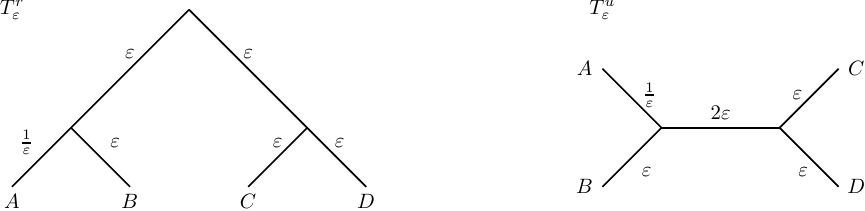}
	\caption{Rooted tree $T_{\varepsilon}^r$ with leaf set $X=\{A,B,C,D\}$ and its unrooted version $T_{\varepsilon}^u$.}
  \label{fig:2}
\end{figure*}

\section{Comparing the rankings induced by the different indices}
\label{sec:2}
Both the original and modified Shapley Value and the unrooted Shapley Value on rooted trees and the original Shapley Value are highly correlated, which was shown in \citep{Fuchs2015} and \citep{Hartmann2013}, respectively. But even if the correlation index goes to 1 as the number of taxa goes to infinity, the indices can still result in different ranking orders of the taxa. 

To illustrate this effect, we considered the phylogenetic tree for European amphibian species presented in \citep{Zupan2014} (available on TreeBASE, accession number: S13561).
We calculated the original Shapley Values, the modified Shapley Values and the unrooted rooted Shapley Values for all $105$ species present in the tree and compared the resulting rankings (see supporting information S1 Text). The rankings are similar, but they are not identical. Consider for example the species \textit{Hyla sarda} and \textit{Bombina bombina}. While the ranking obtained from the original Shapley Value places \textit{Hyla sarda} before \textit{Bombina bombina}, the modified Shapley Value and the unrooted rooted Shapley Value rank \textit{Bombina bombina} higher than \textit{Hyla sarda}. Other discrepancies between the rankings can be found easily. However, all rankings show a similar overall tendency in ranking the species and thus, the distance between the different rankings is low. We have, $d_1^{\ast}(r_{SV},r_{\widetilde{SV}}) \approx 0.01959$ and $d_1^{\ast}(r_{SV},r_{\widehat{SV}}) \approx 0.02612$.
%A full overview of all Shapley values is provided in the supporting information (S1 Text).

In the following we will undertake further statistical analyses to explore this effect and first compare the rankings obtained from the original and modified Shapley Value and afterwards the rankings obtained from the original and unrooted Shapley Value on rooted trees.

For each analysis we have generated some random trees (the details of which will be explained in the subsequent sections). We then calculated the rankings obtained from the different versions of the Shapley Value and counted the number of cases where the rankings are identical. In a subsequent step we calculated the normalized Manhattan distance between the different rankings.

\begin{example}
Consider the rooted tree $T^r$ depicted in Figure \ref{fig:3} and the different versions of the Shapley Value for its taxa as listed in the corresponding table. \\
We obtain the rankings $r_{SV}(A,B,C,D)=(1,3,1,4), \, r_{\widetilde{SV}}(A,B,C,D)=(1,4,2,3)$ and $r_{\widehat{SV}}(A,B,C,D)=(2,4,3,1)$. 
The maximal Manhattan distance between two vectors $r',s'$ of length 4, containing the elements $\{ 1,2,3,4 \}$ is $d_1(r',s')=8$. 
Thus, we have $d_1^*(r_{SV},r_{\widetilde{SV}}) = \frac{3}{8}$, $d_1^*(r_{SV},r_{\widehat{SV}})= \frac{7}{8}$.
\end{example}

%Figure illustrating the possible different ranking order
\begin{figure*}[htbp]
	\centering
	\subfloat[Non-ultrametric rooted tree $T^r$]{\includegraphics[width=0.5\textwidth]{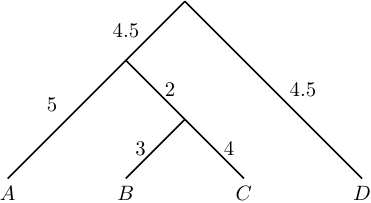} \label{subfig3}}
	\hfil
  	\subfloat[Different versions of the Shapley Value for tree $T^r$]{
  		\begin{tabular}[b]{llll}
  		\hline\noalign{\smallskip}
		 taxon & $SV$ & $\widetilde{SV}$ & $\widehat{SV}$ \\
		\noalign{\smallskip}\hline\noalign{\smallskip}
		 $A$ & 6.5 & 4.125 & 5.58$\overline{3}$ \\
		 $B$ & 5.5 & 3.125 & 4.25 \\
		 $C$ & 6.5 & 3.875 & 4.91$\overline{6}$ \\
		 $D$ & 4.5 & 3.375 & 8.25\\
		\noalign{\smallskip}\hline
		\end{tabular}
		\label{subtable3}}
  \caption{(a) Rooted tree $T^r$ with leaf set $X=\{A,B,C,D\}$ and (b) its different Shapley Values}
  \label{fig:3}
\end{figure*}

We are now in a position to compare the different versions of the Shapley Value.

\subsection{Original vs. modified Shapley Value}
\label{subsec:2.1}
We now compare the ranking order of taxa obtained from the original and the modified Shapley Value. Recall that $SV_{T^r}(a)$ and $\widetilde{SV}_{T^r}(a)$ differ only by the summand $\frac{PD^r(a)}{n}$ (see Equation \eqref{eq:1}). 
Therefore, note that the difference between the two versions of the Shapley Value can be of any size. In particular, if we choose the branch lengths of $T^r$ so long  that $PD^r(a) > n$  for all $a \in X$, the difference between $SV_{T^r}(a)$ and $\widetilde{SV}_{T^r}(a)$ will increase as $n$ increases.

But on the other hand, it can easily be seen that the values are highly correlated, because if $\widetilde{SV}_{T^r}(a)$ gets larger, so does $SV_{T^r}(a)$ (by Equation \eqref{eq:1}). So the values can be made arbitrarily different, but the correlation necessarily affects the ranking of the different taxa. This motivates our analysis of the rankings obtained from the two values instead of the values themselves. We will see in our simulation study that the rankings grow more and more alike with an increasing tree size. 

However, before we present our simulation result, we state the following simple proposition which makes such a simulation unnecessary for ultrametric trees.

\begin{proposition} \label{ultrametric} Let $T^r$ be a rooted binary phylogenetic tree on taxon set $X$. If $T^r$ is ultrametric, then the rankings implied by $SV_{T^r}(a)$ and $\widetilde{SV}_{T^r}(a)$ are identical.
\end{proposition}

\begin{proof} In ultrametric trees, all paths from the root to a leaf are of the same length, say $k$. This implies $PD(a_1) = \ldots =  PD(a_n)=k$ for all leaves $a_1, \ldots, a_n$. Using Equation \eqref{eq:1}, this leads to $SV_{T^r}(a) = \widetilde{SV}_{T^r}(a) + \frac{k}{n}$  for all $a \in X$. Thus, the original and the modified Shapley Values still differ in ultrametric trees, but the difference is the same for all taxa and therefore the ranking order obtained from either of the two indices will be the same.
\end{proof}

So for ultrametric trees, Proposition \ref{ultrametric} shows the equality between the rankings. Therefore, our subsequent analysis is only concerned with non-ultrametric trees. In particular, we want to find out how quickly (relative to the number $n$ of taxa) the rankings of the two indices  $SV_{T^r}(a)$ and $ \widetilde{SV}_{T^r}(a)$ on average coincide for increasing $n$. 

For this analysis, we used \textsf{R} \citep{R}, to be precise the packages ape \citep{ape}, geiger \citep{geiger} and phytools \citep{phytools} to generate a set of 100 random trees for each $n = 10, 20, \ldots, 100$. The trees were generated under a birth-death model with birth rate $\mu=1$ and death rate $\nu=1$. Note that even when a birth-death model is considered, it does not make much sense to consider some of the lineages as extinct, as biodiversity conservation can only aim at present-day species. So in this case, the birth-death model is only used to simulate a suitable tree shape, but not to represent the evolutionary history of the $n$ species under investigation. We always consider all $n$ lineages as extant, and the branch lengths are some measure of difference between the species. A second analysis where we generated random trees based on a uniform distribution of all possible rooted binary tree topologies on $n$ taxa can be found in the supporting information (S1 Text).

After generating the set of 100 random trees for each $n$, we calculated $SV_{T^r}(a)$ and $\widetilde{SV}_{T^r}(a)$ for all $a \in X$. Then we inferred the induced rankings of $SV_{T^r}$ and $\widetilde{SV}_{T^r}$ and calculated the normalized Manhattan distance (see supporting information (S1 Text) for further comparisons with the Kendall tau distance). These values were then summarized in the boxplots depicted in Figure \ref{fig:4} (all Figures were generated using \citep{MATLAB}. On the $x$-axis, we denoted both the number $n$ of taxa as well as the absolute counts (in brackets) of rankings which were identical.

The number of identical rankings decreased from 20 for 10 taxa down to 0 for 30 or more taxa.
Thus, the rankings obtained from the two versions of the Shapley Value tend to differ when the tree becomes large despite the high correlation. 

However, the amount of dissimilarity between the rankings decreases with an increasing tree size, as both the variability of the obtained distances within a tree set of fixed size (as depicted by the boxplots Figure \ref{fig:4}) and the mean distance (see Figure \ref{fig:4}) decrease with a growing number of taxa.

In summary we can say that the original and modified Shapley Value are more likely to result in an identical ranking order for small trees than for large trees. However, when they do not result in exactly the same ranking, the rankings are on average `more similar' for large trees than for small trees (in terms of the normalized Manhattan distance).

%Figure Random Trees and birth death trees (ManhattanDistance)
\begin{figure*}[htbp]
	\centering
		\includegraphics[width=0.6\textwidth]{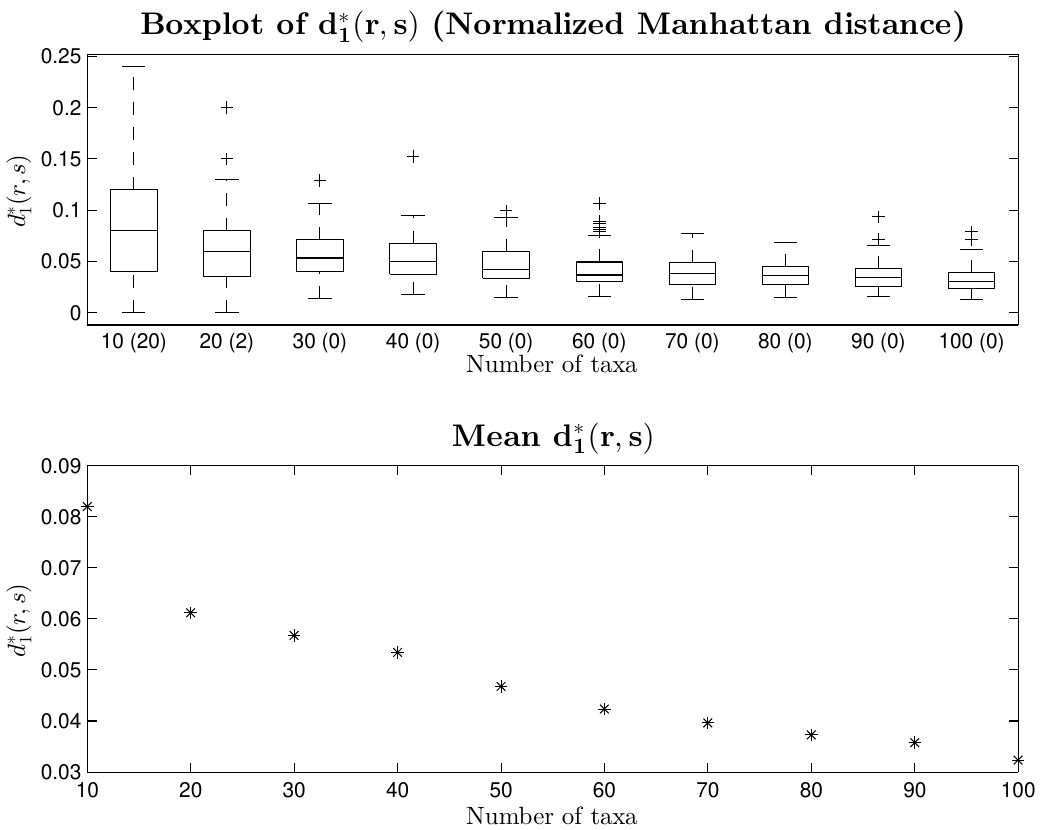}
	\caption{Boxplots and plot of the mean of the normalized Manhattan distance between rankings obtained from the modified and the original Shapley Value for 100 random birth-death trees with $\mu = \nu = 1$ of size $n= 10, 20, \ldots, 100$, respectively. The numbers in round brackets in the boxplots denote the number of identical rankings, e.g., in case of the birth-death trees with 10 taxa, we counted 20 identical and 80 dissimilar rankings.}
	\label{fig:4}
\end{figure*}

\subsection{Original Shapley Value vs. unrooted rooted Shapley Value}  
\label{subsec:2.2}
We now compare the Original Shapley Value and the unrooted rooted Shapley Value. In this case, we have to consider both ultrametric and non-ultrametric trees, because in both cases the ranking order may differ.
In the non-ultrametric case we used the set of birth-death trees from above for the analysis, while we generated a set of random trees under a Yule-model (i.e., a pure birth process) for the ultrametric case.

Again, we counted the number of identical rankings and calculated the normalized Manhattan distance (see Figure \ref{fig:5}).

We first notice that the number of cases where the rankings obtained from the original and unrooted rooted Shapley Value are identical is higher for Yule trees (ultrametric) than for birth-death trees (non-ultrametric). 
Similarly, the variability of both the normalized Manhattan distance itself, is smaller in Yule trees than in birth-death trees. Still, it decreases with an increasing number of taxa.
Notice, however, the difference in the scaling of the $y$-axis for Yule trees and birth-death trees. In all cases, the mean normalized Manhattan distance between rankings obtained from the original Shapley Value and rankings obtained from the unrooted rooted Shapley Value are approximately ten times smaller in Yule trees than in birth-death trees.

In summary we can say that the number of times where we obtain identical rankings from the two versions of the Shapley Value decreases with a growing number of taxa. At the same time, the normalized distances between the rankings decrease. Thus, in those cases where the rankings are not exactly identical, they tend to be `more similar' for large trees than for small trees. This reflects the effect which we have already observed for the rankings obtained from the modified and the original Shapley Value.
However, the impact of using the unrooted rooted Shapley Value as opposed to the original Shapley Value is higher for non-ultrametric trees than for ultrametric trees. The two indices can result in different rankings in both cases, but they are less different in the ultrametric case. Remember that the original and modified Shapley Value always lead to the same ranking order of taxa for ultrametric trees, which the original and unrooted rooted Shapley Value may not do. Still, ultrametricity seems to imply that it is less important which version of the Shapley Value we use to obtain a ranking of taxa.

%Figure Birth-death trees and Yule-trees (Manhattan Distance)
\begin{figure*}[htbp]
	\centering
	\subfloat[Random Yule trees (ultrametric) ]{
		\includegraphics[width=0.6\textwidth]{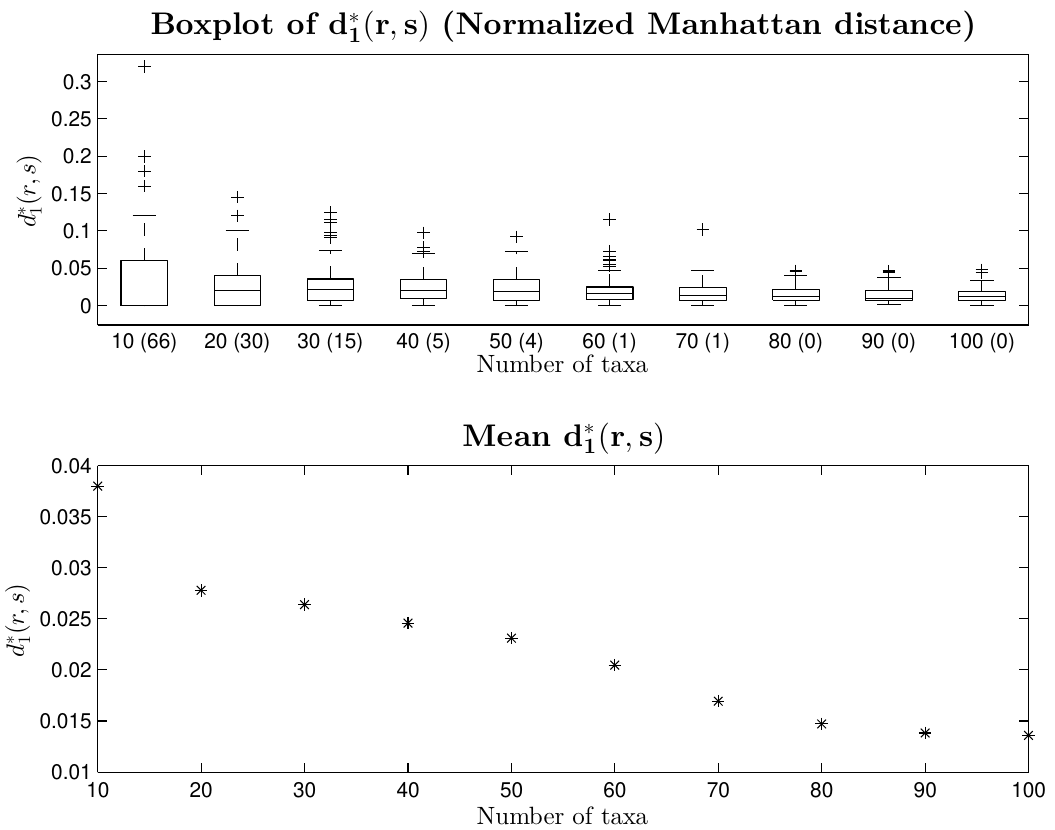}
		\label{subfig:5.1}}
	\vfill
	\subfloat[Random birth-death trees with $\mu$ = $\nu$ = 1 (non-ultrametric)]{
		\includegraphics[width=0.6\textwidth]{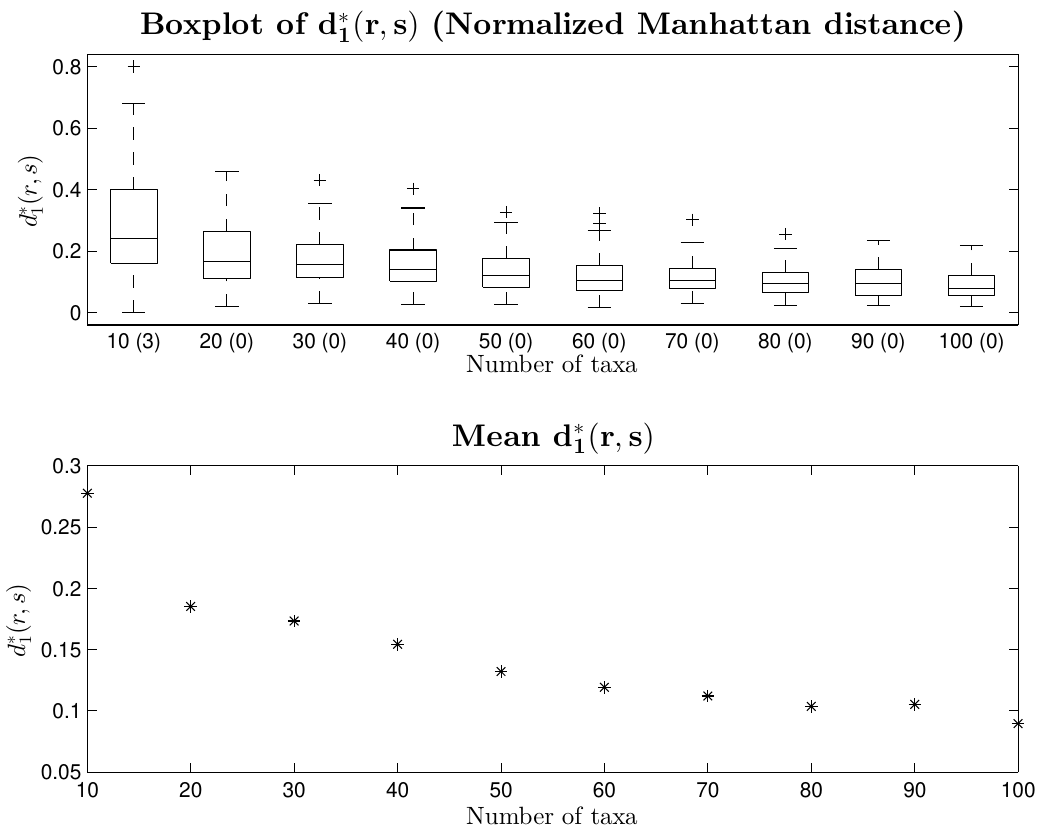}
		\label{subfig:5.2}}
	\caption{Boxplots and plot of the mean of the normalized Manhattan distance between rankings obtained from the original Shapley Value and the unrooted Shapley Value on rooted trees for 100 trees of size $n= 10, 20, \ldots,100$, respectively. The numbers in round brackets in the boxplots denote the number of identical rankings, e.g., in case of the birth-death trees with 10 taxa, we counted 3 identical and 97 dissimilar rankings.}
	\label{fig:5}
\end{figure*}

\section{Software}
\label{sec:3}
In order to calculate the different versions of the Shapley Value, we developed a program called FairShapley, which is available from \\ http://www.mareikefischer.de/Software/FairShapley.zip. A readme file containing a short manual is available at http://www.mareikefischer.de/Software/ FairShapleyREADME.txt. The tool is written in the programming language Perl and uses modules from BioPerl \citep{Bioperl}.
In contrast to existing tools for the Fair Proportion Index and the Shapley Value \citep{Vos2011}, which use the unrooted version of the Shapley Value (regardless of the tree being rooted or not) and arbitrarily root unrooted trees in case of the Fair Proportion Index, our program explicitly distinguishes between rooted and unrooted trees. It allows for the computation of both the original Shapley Value (which coincides with the Fair Proportion Index if the tree is rooted, as we pointed out earlier), the modified Shapley Value and the unrooted Shapley Value on rooted trees. 
The program takes trees represented in the so-called Newick format (cf. \citet{newick}) as an input. This format uses brackets, and two closely related species are grouped closely together. Moreover, a binary tree has two entries at each bracket level -- except for unrooted trees, which have three entries at the uppermost level. For example, the trees in Figure \ref{fig:1} can be denoted by $T^r=(((A:1,B:1),C:2):1,D:3)$ and $T^u=((A:1,B:1):1,C:2,D:4)$. Note that the numbers denote the edge lengths, e.g., $A:1$ means that the edge leading to leaf $A$ has length $1$. The program outputs the ranking order of taxa obtained from any version of the Shapley Value.

\section{Conclusions}
\label{sec:4}
In this paper, we summarized the different versions of the Shapley Value which can be found in the literature and which have different relationships with one another as well as with the frequently used Fair Proportion Index. We also showed that even though the different definitions are all highly correlated, the rankings they induce are hardly ever identical when the number $n$ of taxa under investigation is large. But the difference between the different rankings converges to zero as $n$ grows. 
We have chosen the Manhattan distance to measure the difference between the rankings obtained from different biodiversity indices, because the Manhattan distance can deal with ties that may occur in the rankings. Additionally, we have used the Kendall tau distance, which is aiming at rankings, but cannot deal with ties in a second analysis (see supporting information (S1 Text)).

In total, we have seen that both investigated distance measures lead to similar results. We have also seen that ultrametricity of the underlying tree makes the rankings of  all indices more similar -- in the case of the original and the modified Shapley Values, we even get equality for ultrametric trees. However, for non-ultrametric trees, the rankings tend to differ more, and even though the different versions of the Shapley Value are highly correlated, surprisingly the probability of getting two identical rankings from these different versions for a given tree decreases as the number $n$ of taxa increases. Yet the normalized differences tend to zero, and the variability of the differences between the values gets smaller for increasing $n$. 

We conclude that all biodiversity prioritization indices discussed in current literature, namely the Fair Proportion Index as well as all versions of the Shapely Value, tend to give similar results, particularly if the number of species under consideration is large. As the Fair Proportion Index can be calculated most easily, we therefore think its wide use is justified. However, as the probability of getting identical rankings from different values is small for a large species set, we suggest that for biological data and real conservation decisions, more than one index should be taken into account. 

It would be of interest to see if the differences in the rankings, which we observed for different kinds of random trees and the Amphibian dataset \citep{Zupan2014}, also occur frequently on other phylogenetic trees reconstructed from real data. This is another possible area for future research.

\section*{Supporting Information}
\textbf{S1 Text. Supporting information file that contains all versions of the Shapley Value for the European amphibian tree presented in \citep{Zupan2014}, and additional results of the simulation study.}

\section*{Acknowledgements}
\noindent
We want to thank an anonymous reviewer and Volkmar Liebscher for helpful comments on an earlier version of the manuscript.

\section*{References}
\bibliographystyle{model1-num-names}\biboptions{authoryear}
\bibliography{Sources}   % name your BibTeX data base

\end{document}